\documentclass[runningheads]{llncs}
\usepackage[english]{babel}
\usepackage{microtype}
\usepackage{tikz} 
\usepackage{xspace} 
\usepackage{amssymb} 
\usepackage{amsmath}
\usepackage[makeroom]{cancel} 
\usepackage{algorithm}
\usepackage[noend]{algpseudocode}
\usepackage{enumerate}
\usepackage{hyperref}
\usepackage{multirow}
\usepackage{ifthen}
\usepackage{todonotes}
\usetikzlibrary{arrows}

\def\G{\mathcal{G}}

\def\R{\mathbb{R}}

\def\SMALGO{\textsc{Smalgo}\xspace}
\def\SMO{\mbox{\textsc{Smalgo}}\xspace}
\def\SMI{\mbox{\textsc{Smalgo}-I}\xspace}
\def\SMII{\mbox{\textsc{Smalgo}-II}\xspace}
\def\PGRAPH{\mbox{$P$-graph}\xspace}
\def\TGRAPH{\mbox{$T$-graph}\xspace}

\def\DMASKS{\mbox{D-masks}\xspace}
\def\CS{Cross Sampling\xspace}
\def\BCS{Backward Cross Sampling\xspace}
\def\SHIFTAND{\mbox{Shift-And}\xspace}
\def\SHIFTOR{\mbox{Shift-Or}\xspace}
\def\PMASK{\mbox{P-mask}\xspace}
\def\PMASKS{\mbox{P-masks}\xspace}

\def\GTM{graph theoretic model\xspace}
\def\GSM{GSM\xspace}

\def\SM{Swap Matching\xspace}
\def\O{O}
\def\AND{\mathbin{\&}}
\def\OR{\mid}

\def\BMA{BMA\xspace}

\def\TH{\text{-th}}
\def\RD{\text{-rd}}

\def\S[#1]{S_{#1}}
\def\SP[#1]{\widehat S_{#1}}
\def\SS[#1,#2,#3]{{#1}_{[#2,#3]}}
\def\ALF{\Sigma}
\def\LAB{\sigma}
\def\DD[#1]{D_{#1}}
\def\D[#1,#2]{D^{#1}_{#2}}
\def\DR[#1,#2,#3]{\mathit{D#1}^{#2}_{#3}}
\def\R[#1,#2,#3]{R{#3}^{#1}_{#2}}
\def\RR[#1,#2,#3]{{R#1}^{#2}_{#3}}
\def\SHORT[#1,#2,#3]{{#1}^{#2}_{#3}}
\def\M[#1,#2]{m_{#1,#2}}
\def\G[#1,#2]{\mathcal{#1}_{#2}}
\def\S[#1]{S_{#1}}
\def\I[#1]{I^{#1}}
\def\MASK[#1,#2,#3]{\mathit{#1}^{#2}_{#3}}

\usetikzlibrary{calc, shapes, backgrounds, patterns, arrows, decorations,
                decorations.pathreplacing, decorations.pathmorphing, fit}

\tikzset{%
    ->,>=stealth',node distance = 1cm,shorten >= 2pt%
    ,main/.style={fill = none, draw = none, font=\texttt\sffamily}
    ,red/.style={draw = red,font=\sffamily\bfseries}
    ,blue/.style={draw = blue,circle,font=\sffamily\bfseries}
    ,red/.style={draw = red}
}

\newcommand*{\ShortArrow}{7pt}

\newcommand{\POverWcc}{{\lceil\frac{p}{w}\rceil}}
\newtheorem{observation}[theorem]{Observation}

\title{A Simple Streaming Bit-parallel Algorithm\\ for Swap Pattern Matching\thanks{An extended abstract of this work appeared
in the Proceedings of the 7th International Conference on Mathematical Aspects of Computer and Information Sciences, MACIS 2017~\cite{BSV17}.}}
\author{V\'aclav Bla\v zej%
\thanks{Supported by the OP VVV MEYS funded project CZ.02.1.01/0.0/0.0/16\_019/0000765 ``Research Center for Informatics''
and by the SGS CTU project SGS17/209/OHK3/3T/18.}
\and Ond\v rej Such\'y%
\thanks{Supported by grant 17-20065S of the Czech Science Foundation.}
\and Tom\'a\v s Valla%
\thanks{Supported by the Centre of Excellence -- Inst.\ for Theor.\ Comp.\ Sci. 79 (project P202/12/G061 of the Czech Science Foundation.)}
}
\institute{Faculty of Information Technology, Czech Technical University in Prague,\\ Prague, Czech Republic}

\begin{document}

\maketitle


\begin{abstract}
The pattern matching problem with swaps is to find all occurrences of a pattern in a text while allowing the pattern to swap adjacent symbols.
The goal is to design fast matching algorithm that takes advantage of the bit parallelism of bitwise machine instructions
and has only streaming access to the input.
We introduce a new approach to solve this problem based on the graph theoretic model and compare its performance to previously known algorithms.
We also show that an approach using deterministic finite automata cannot achieve similarly efficient algorithms.
Furthermore, we describe a fatal flaw in some of the previously published algorithms based on the same model.
Finally, we provide experimental evaluation of our algorithm on real-world data.
 \end{abstract}

\section{Introduction}

In the \emph{Pattern Matching problem with Swaps} (\emph{\SM}, for short),
the goal is to find all occurrences of any \emph{swapped version} of a pattern~$P$ in a text~$T$, where~$P$ and~$T$ are strings of length~$p$ and~$t$ over an alphabet~$\ALF$, respectively.
By the swapped version of a pattern~$P$ we mean a string of symbols created from~$P$ by swapping adjacent symbols while ensuring that each symbol is swapped at most once (see Section~\ref{chap:notions} for formal definitions).
The solution of \SM is a~set of indices which represent where occurrences swapped version of~$P$ in~$T$ begin.
\SM is intensively studied due to its use in practical applications such as text and music retrieval, data mining, network security and biological computing~\cite{AIJR08ImpSwapMatch}.

The swap of two consecutive symbols is one of the most typical typing errors. It also represent a simpler version of swaps that appear in nature. In particular, the phenomenon of swaps occurs in gene mutations and duplications such as in the region of human chromosome 5 that is implicated in the disease called spinal muscular Atrophy, a common recessive form of muscular dystrophy~\cite{Lewin95}.
While the biological swaps occur at a gene level and have several additional constraints and characteristics, which make the problem  much more difficult, they do serve as a convincing pointer to the theoretical study of swaps as a natural edit operation for the approximation metric~\cite{Amir2000247}. Indeed Lowrance and Wagner~\cite{WagnerL1975} suggested to add the swap operation when considering the edit distance of two strings.

\SM was introduced in 1995 as an open problem in non-standard string matching~\cite{origin}.
The first result was reported by Amir et al.~\cite{Amir2000247} in 1997, who provided an $\O(t p^{\frac{1}{3}}\log p)$-time solution for alphabets of size~$2$, while also showing that alphabets of size exceeding~$2$ can be reduced to size~$2$ with a little overhead.
Amir et al.~\cite{Amir1998125} came up with solution with $\O(t \log^2 p)$ time complexity for some very restrictive cases.
Several years later Amir et al.~\cite{Amir200357} showed that \SM
can be solved by an algorithm for the overlap matching achieving the running
time of $\O(t\log p \log |\ALF|)$. This algorithm as well as all the previous
ones is based on fast Fourier transformation (FFT).

In 2008 Iliopoulos and Rahman~\cite{newModel} introduced a new graph theoretic approach to model the \SM problem and came up with the first efficient solution to \SM without using FFT (we show it to be incorrect).
Their algorithm based on bit parallelism runs in $\O((t+p)\log p)$ time if the pattern length is similar to the word-size of the target machine.
One year later Cantone and Faro~\cite{CS} presented a dynamic programming algorithm named \CS  solving \SM in~$\O(t)$ time and~$\O(|\ALF|)$ space, assuming that the pattern length is similar to the word-size in the target machine.
In the same year Campanelli et al.~\cite{BCS} enhanced the \CS algorithm using notions from Backward directed acyclic word graph matching algorithm and named the new algorithm \BCS.
This algorithm also assumes short pattern length.
Although \BCS has $\O(|\ALF|)$ space and $\O(tp)$ time complexity,
which is worse than that of \CS, it improves the real-world performance.

In 2013 Faro~\cite{automata} presented a new model to solve \SM using reactive automata
and also presented a new algorithm with $\O(t)$ time complexity assuming short patterns.
The same year Chedid~\cite{Chedid} reformulated the dynamic programming solution by Cantone and Faro~\cite{CS}
which results in more intuitive algorithms.
In 2014 a minor improvement by Fredriksson and Giaquinta~\cite{Fredriksson2014} appeared,
yielding slightly (at most factor $|\ALF|$) better asymptotic time complexity (and also slightly worse space complexity) for special cases of patterns.
The same year Ahmed et al.~\cite{revisited} took ideas of the algorithm by Iliopoulos and Rahman~\cite{newModel} and devised two algorithms named \SMI and \SMII which both run in $\O(t)$ for short patterns, but bear the same error as the original algorithm.

Another remarkable effort related to \SM is to actually count the number of swaps needed to match the pattern at the location~\cite{Amir2002}. This is more often studied
with an extra operation of character change allowed~\cite{Amir,Dombb2010,Lipsky2010}.

\subsubsection*{Our Contribution.}
We design a simple algorithm which solves the \SM problem.
The goal is to design a streaming algorithm, which is given one symbol per each execution step until
the end-of-input arrives, and thus does not need access to the whole input.
This algorithm has $O(\POverWcc(|\ALF|+t) + p)$ time and $O(\POverWcc|\ALF|)$ space complexity where~$w$ is the word-size of the machine.
We would like to stress that our solution, as based on the graph theoretic approach, does not use FFT.
Therefore, it yields a much simpler non-recursive algorithm allowing bit parallelism and is not suffering from the disadvantages of the convolution-based methods.
While our algorithm matches the best asymptotic complexity bounds of the previous results~\cite{CS,Fredriksson2014} (up to a $|\ALF|$ factor), we believe that its strength lies in the applications where the alphabet is small and the pattern length is at most the word-size, as it can be implemented using only $7+|\ALF|$ CPU registers and few machine instructions.
This makes it practical for tasks like DNA sequences scanning.
Also, as far as we know, our algorithm is currently the only known streaming algorithm for the swap matching problem.

We continue by proving that any deterministic finite automaton that solves \SM has number of states exponential in the length of the pattern. 


We also describe the \SMO (swap matching algorithm) by Iliopoulos and Rahman \cite{newModel} in detail.
Unfortunately, we have discovered that \SMO and derived algorithms contain a flaw which cause false positives to appear.
We have prepared implementations of \SMI, Cross Sampling, Backward Cross Sampling and our own algorithm,
measured the running times and the rate of false positives for the \SMI algorithm.
All of the sources are available for download.\footnote{\url{http://users.fit.cvut.cz/blazeva1/gsm.html}}

This paper is organized as follows.
First we introduce all the basic definitions, and also recall the graph theoretic model introduced in~\cite{newModel} and its use for matching in Section~\ref{chap:notions}.
In Section~\ref{chap:our} we show our algorithm for \SM problem and follow it in Section~\ref{sec:limits} with the proof that \SM cannot be solved efficiently by deterministic finite automata.
Then we describe the \SMO algorithms in detail in Section~\ref{chap:smalgo} and finish with the experimental evaluation of the algorithms in Section~\ref{chap:experiments}.


\section{Basic Definitions and the Graph Theoretic Model}
In this section we state the basic definitions, present the graph theoretic model and show a basic algorithm that solves \SM using the model.
\label{chap:notions}
\subsection{Notations and Basic Definitions}

We use the word-RAM as our computational model. That means we have access to memory cells of fixed capacity~$w$ (e.g., 64 bits).
A standard set of arithmetic and bitwise instructions include And ($\AND$), Or ($\OR$), Left bitwise-shift ($\mbox{LShift}$ or $\ll 1$) and Right bitwise-shift ($\mbox{RShift}$ or $\gg 1$).
Each of the standard operations on words takes single unit of time.
In order to compare to other existing algorithms, which are not streaming, we
define the access to the input in a less restrictive way -- the input is read
from a read-only part of memory and the output is written to a write-only part
of memory.
However, it will be easy to observe that our algorithm accesses the input sequentially.
We do not include the input and the output into the space complexity analysis.

A \emph{string}~$S$ over an alphabet~$\ALF$ is a finite sequence of symbols from~$\ALF$ and~$|S|$ is its  length.
By~$\S[i]$ we mean the~$i$-th symbol of~$S$ and we define a \emph{substring} $\SS[S,i,j]=\S[i] \S[i+1] \dots \S[j]$ for $1 \le i \le j \le |S|$, and \emph{prefix} $\SS[S,1,i]$ for $1 \le i \le |S|$.
String~$P$ \emph{prefix matches} string~$T$~$k$~symbols on position~$i$ if $P_{[1,k]}=T_{[i,i+k-1]}$.

Next we formally introduce a swapped version of a string.

\begin{definition}[Campanelli et al.~\cite{BCS}]\label{def:SMP}
   A \emph{swap permutation} for~$S$ is a permutation $\pi : \{ 1, \dots, n \} \to \{ 1, \dots, n \}$, where $n= |S|$, such that:
  \begin{enumerate}[(i)]
    \item if $\pi(i) = j$ then $\pi(j) = i$ (symbols at positions~$i$ and~$j$ are swapped),
    \item for all $i,\pi(i)\in\{i-1,i,i+1\}$ (only adjacent symbols are swapped),
    \item if $\pi(i) \ne i$ then $S_{\pi(i)} \ne S_i$ (identical symbols are not swapped).
  \end{enumerate}
For a string~$S$ a \emph{swapped version}~$\pi(S)$ is a string $\pi(S) = S_{\pi(1)} S_{\pi(2)} \dots S_{\pi(n)}$ where~$\pi$ is a swap permutation for~$S$.
\end{definition}


Now we formalize the version of matching we are interested in.

\begin{definition}
  Given a text $T = T_1 T_2 \dots T_t$ and a pattern $P = P_1 P_2 \dots P_p$, the pattern~$P$ is said to \emph{swap match}~$T$ at location~$i$ if there exists a swapped version~$\pi(P)$ of~$P$ that matches~$T$ at location~$i$, i.e., $\pi(P) = \SS[T,i,i+p-1]$.
\end{definition}

\subsection{A Graph Theoretic Model}
\label{sec:model}
The algorithms in this paper are based on a model introduced by Iliopoulos and Rahman~\cite{newModel}.
In this section we briefly describe this model.

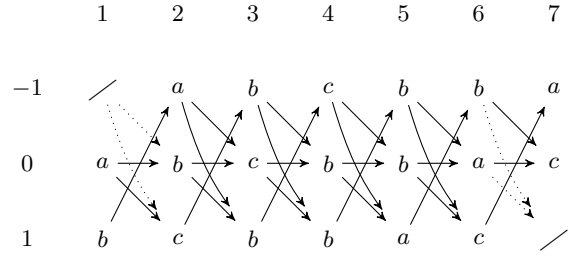
\begin{figure}[t]
  \centering
  \begin{tikzpicture}
    \node[main] (00) {};
    \node[main] (01) [below of = 00] {$-1$};
    \node[main] (02) [below of = 01] {$0$};
    \node[main] (03) [below of = 02] {$1$};
    \node[main] (10) [right of = 00] {$1$};
    \node[main] (20) [right of = 10] {$2$};
    \node[main] (30) [right of = 20] {$3$};
    \node[main] (40) [right of = 30] {$4$};
    \node[main] (50) [right of = 40] {$5$};
    \node[main] (60) [right of = 50] {$6$};
    \node[main] (70) [right of = 60] {$7$};

    \node[main] (12) [right of = 02] {$a$};
    \node[main] (22) [right of = 12] {$b$};
    \node[main] (32) [right of = 22] {$c$};
    \node[main] (42) [right of = 32] {$b$};
    \node[main] (52) [right of = 42] {$b$};
    \node[main] (62) [right of = 52] {$a$};
    \node[main] (72) [right of = 62] {$c$};

    \node[main] (11) [above of = 12] {$\cancel{~ ~}$};
    \node[main] (21) [above of = 22] {$a$};
    \node[main] (31) [above of = 32] {$b$};
    \node[main] (41) [above of = 42] {$c$};
    \node[main] (51) [above of = 52] {$b$};
    \node[main] (61) [above of = 62] {$b$};
    \node[main] (71) [above of = 72] {$a$};

    \node[main] (13) [below of = 12] {$b$};
    \node[main] (23) [below of = 22] {$c$};
    \node[main] (33) [below of = 32] {$b$};
    \node[main] (43) [below of = 42] {$b$};
    \node[main] (53) [below of = 52] {$a$};
    \node[main] (63) [below of = 62] {$c$};
    \node[main] (73) [below of = 72] {$\cancel{~ ~}$};

    \path[every node/.style={font=\sffamily\small}]
      (12) edge (22)
      (22) edge (32)
      (32) edge (42)
      (42) edge (52)
      (52) edge (62)
      (62) edge (72)

      (12) edge (23)
      (22) edge (33)
      (32) edge (43)
      (42) edge (53)
      (52) edge (63)
      (62) edge [dotted] (73)

      (13) edge (21)
      (23) edge (31)
      (33) edge (41)
      (43) edge (51)
      (53) edge (61)
      (63) edge (71)

      (11) edge [dotted] (22)
      (21) edge (32)
      (31) edge (42)
      (41) edge (52)
      (51) edge (62)
      (61) edge (72)

      (11) edge [shorten >= \ShortArrow , bend right = 10, dotted] (23)
      (21) edge [shorten >= \ShortArrow , bend right = 10] (33)
      (31) edge [shorten >= \ShortArrow , bend right = 10] (43)
      (41) edge [shorten >= \ShortArrow , bend right = 10] (53)
      (51) edge [shorten >= \ShortArrow , bend right = 10] (63)
      (61) edge [shorten >= \ShortArrow , bend right = 10, dotted] (73)
    ;
  \end{tikzpicture}
  \caption{\PGRAPH $\G[P,P]$ for the pattern $P = abcbbac$}
  \label{fig:example}
\end{figure}


%
For a pattern~$P$ of length~$p$ we construct a labeled graph $\G[P,P]=(V, E, \LAB)$  with vertices~$V$, edges~$E$, and a vertex labeling function $\LAB : V \to \ALF$ (see Fig.~\ref{fig:example} for an example).
Let $V=V'\setminus \{ \M[-1,1] , \M[1,p] \} $ where
$V' = \{\M[r,c]\mid r \in \{-1,0,1\}, c \in \{1,2,\dots,p\}\}$. 
For $\M[r,c] \in V$ we set $\LAB(\M[r,c])=P_{r+c}$. 
Each vertex~$\M[r,c]$ is identified with an element of a $3 \times p$ grid.
We set $E' := E'_1 \cup E'_2 \cup \dots \cup E'_{p-1}$, where
$E'_j := \{(m_{k,j},m_{i,j+1}) \mid k \in \{-1,0\},i \in \{0, 1\}\}\cup\{(m_{1,j},m_{-1,j+1})\}$,
and let $E = E' \cap V \times V$.
We call $\G[P,P]$ the \emph{\PGRAPH}.
Note that $\G[P,P]$ is directed acyclic graph, $|V(\G[P,P])| = 3p-2$, and $|E(\G[P,P])| = 5(p-1)-4$.

The idea behind the construction of $\G[P,P]$ is as follows.
We create vertices~$V'$ and edges~$E'$ which represent every swap pattern without unnecessary restrictions (equal symbols can be swapped).
We remove vertices~$\M[-1,1]$ and~$\M[1,p]$ which represent symbols from invalid indices~$0$ and~$p+1$.

The \PGRAPH now represents all possible swap permutations of the pattern~$P$ in the following sense.
Vertices~$\M[0,j]$ represent ends of prefixes of swapped version of the pattern which end by a non-swapped symbol.
Possible swap of symbols~$P_j$ and $P_{j+1}$ is represented by vertices $\M[1,j]$ and $\M[-1,j+1]$.
Edges represent symbols which can be consecutive.
Each path from column~$1$ to column~$p$ represents a~swap pattern and each swap pattern is represented this way.


\begin{definition}
  For a given~$\ALF$-labeled directed acyclic graph $G=(V,E,\LAB)$ vertices $s,e \in V$ and a directed path $f=v_1,v_2,\dots,v_k$ from $v_1=s$ to $v_k=e$, we call $S = \LAB(f) = \LAB(v_1) \LAB(v_2) \dots \LAB(v_k) \in \ALF^*$ a \emph{path string} of~$f$.
\end{definition}

\subsection{Using Graph Theoretic Model for Matching}

\label{sec:basicMatching}
In this section we describe an algorithm called \emph{Basic Matching Algorithm} (\BMA) which can determine whether there is a match of pattern~$P$ in text~$T$ on a position~$k$ using any graph model which satisfies the following conditions.
\begin{itemize}
  \item It is a directed acyclic graph,
  \item $V = V_1 \uplus V_2 \uplus \dots \uplus V_p$ (we can divide vertices to columns),
  \item $E \subseteq \{ (u, w) \mid u \in V_i , w \in V_{i+1} , 1 \le i < p \}$ (edges lead to next column).
\end{itemize}

Let $Q_0 = V_1$ be the \emph{starting vertices} and $F = V_p$ be the \emph{accepting vertices}.
\BMA is designed to run on any graph which satisfies these conditions.
Since \PGRAPH satisfies these assumptions we can use \BMA for~$\G[P,P]$.

\begin{algorithm}[t]
  \caption{The basic matching algorithm (\BMA)} \label{alg:BMA}
  \begin{algorithmic}[1]
  \Statex \textbf{Input}: Labeled directed acyclic graph $G=(V,E,\LAB)$, set $Q_0 \subseteq V$ of starting vertices, set $F \subseteq V$ of accepting vertices, text $T$, and position $k$.
  \State Let $\DD[1]' := Q_0$. \label{step:init}
  \For{$i = 1,2,3,\dots,p$}  \label{step:repeatstep}
    \State Let $\DD[i] := \{ x \mid x \in \DD[i]', \LAB(x) = T_{k+i-1} \}$.    \label{step:filter}
    \If {$\DD[i] = \emptyset$} finish. \EndIf    \label{step:end}
    \If {$\DD[i] \cap F \ne \emptyset$} we have found a match and finish. \EndIf    \label{step:match}
    \State Define the next iteration set $\DD[i+1]'$ as vertices which are successors of~$\DD[i]$, i.e.,
    \Statex \hspace{1cm}$\DD[i+1]' := \{ d \in V(\G[P,P]) \mid (v,d) \in E(\G[P,P]) \text{ for some } v \in \DD[i] \}$.     \label{step:propagate}
  \EndFor
 \end{algorithmic}
\end{algorithm}

The algorithm runs as follows (see also Algorithm~\ref{alg:BMA}).
We initialize the algorithm by setting $\DD[1]' := Q_0$ (Step~\ref{step:init}).
$\DD[1]'$ now holds information about vertices which are the end of some path~$f$ starting in~$Q_0$ for which~$\LAB(f)$ possibly prefix matches~$1$ symbol of $\SS[T,k,k+p-1]$.
To make sure that the path~$f$ represents a prefix match we need to check whether the label of the last vertex of the path~$f$ matches the symbol~$T_k$ (Step~\ref{step:filter}).
If no prefix match is left we did not find a match (Step~\ref{step:end}).
If some prefix match is left we need to check whether we already have a complete match (Step~\ref{step:match}).
If the algorithm did not stop it means that we have some prefix match but it is not a complete match yet.
Therefore we can try to extend this prefix match by one symbol (Step~\ref{step:propagate}) and check whether it is a valid prefix match (Step~\ref{step:filter}).
Since we extend the matched prefix in each step, we repeat these steps until the prefix match is as long as the pattern (Step~\ref{step:repeatstep}).

Having vertices in sets is not handy for computing so we present another way to describe this algorithm. We use their characteristic vectors instead.

\begin{definition}
  A Boolean labeling function $\I[] : V \to \{ 0,1 \}$ of vertices of~$\G[P,P]$ is called a \emph{prefix match signal}.
  \label{def:prefixmatchsignal}
\end{definition}

The algorithm can be easily divided into \emph{iterations} according to the value of~$i$ in Step~\ref{step:repeatstep}.
We denote the value of the prefix match signal in~$j\TH$ iteration as~$\I[j]$ and we define the following operations:
\begin{itemize}
  \item \emph{propagate signal along the edges}, is an operation which sets $\I[j](v) := 1$ if and only if there exists an edge $(u,v) \in E$ with $\I[j-1](u) = 1$,
  \item \emph{filter signal by a symbol $x \in \ALF$}, is an operation which sets $\I[](v) := 0$ for each~$v$ where $\LAB(v) \ne x$,
  \item \emph{match check}, is an operation which checks whether there exists $v \in F$ such that $I(v) = 1$ and if so reports a match.
\end{itemize}
With these definitions in hand we can describe \BMA in terms of prefix match signals as Algorithm~\ref{alg:pms}.
See Fig.~\ref{fig:examplegraph}
for an example of use of \BMA to figure out whether $P = acbab$ swap matches $T = babcabc$ at a position~$2$.

\begin{algorithm}[t]
 \caption{\BMA in terms of prefix match signals}\label{alg:pms}
  \begin{algorithmic}[1]
  \State Let $\I[0](v) := 1$ for each $v \in Q_0$ and $\I[0](v) := 0$ for each $v \notin Q_0$.
  \For {$i = 0,1,2,3,\dots,p-1$}
    \State Filter signals by a symbol $T_{k+i}$.
    \If{$\I[i](v) = 0$ for every $v \in \G[P,P]$} finish. \EndIf
    \If{$\I[i](v) = 1$ for any $v \in F$} we have found a match and finish. \EndIf
    \State Propagate signals along the edges.
  \EndFor
  \end{algorithmic}
\end{algorithm}

\begin{figure}[t]
  \centering
  \begin{tikzpicture}
    \node[main] (12)                 {$a^1$};
    \node[main] (22) [right of = 12] {$c$};
    \node[main] (32) [right of = 22] {$b$};
    \node[main] (42) [right of = 32] {$a^4$};
    \node[main] (52) [right of = 42] {$b^5$};

    \node[main] (21) [above of = 22] {$a$};
    \node[main] (31) [above of = 32] {$c^3$};
    \node[main] (41) [above of = 42] {$b$};
    \node[main] (51) [above of = 52] {$a$};

    \node[main] (13) [below of = 12] {$c$};
    \node[main] (23) [below of = 22] {$b^2$};
    \node[main] (33) [below of = 32] {$a$};
    \node[main] (43) [below of = 42] {$b$};

    \path[every node/.style={font=\sffamily\small}]
    (12) edge (22)
    (22) edge (32)
    (32) edge (42)
    (42) edge [dashed] (52)

    (12) edge [dashed] (23)
    (22) edge (33)
    (32) edge (43)

    (13) edge (21)
    (23) edge [dashed] (31)
    (33) edge (41)
    (43) edge (51)

    (21) edge (32)
    (31) edge [dashed] (42)
    (41) edge (52)

    (21) edge [shorten >= \ShortArrow , bend right = 10] (33)
    (31) edge [shorten >= \ShortArrow , bend right = 10] (43)
    ;
  \end{tikzpicture}
  \caption{\BMA of $\SS[T,2,6] = abcab$ on a \PGRAPH of the pattern $P = acbab$. The prefix match signal propagates along the dashed edges. Index~$j$ above a vertex~$v$ represent that $\I[j](v) = 1$, otherwise $\I[j](v) = 0$.}
  \label{fig:examplegraph}
\end{figure}
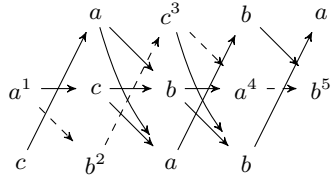


\subsection{Shift-And Algorithm}
\label{sec:shift-and}
The following description is based on~\cite[Chapter 5]{handbook} describing the \SHIFTOR algorithm.

For a pattern~$P$ and a text~$T$ of length~$p$ and~$t$, respectively, let~$\R[,,]$ be a bit array of size~$p$ and~$\R[j,,]$ its value after text symbol~$T_j$ has been processed.
It contains information about all matches of prefixes of~$P$ that end at the position~$j$ in the text.
For $1 \le i \le p$, $\R[j,i,] = 1$ if $\SS[P,1,i] = \SS[T,j-i+1,j]$
and 0 otherwise.
The vector $\R[j+1,,]$ can be computed from~$\R[j,,]$ as follows.
For each positive~$i$ we have $\R[j+1,i+1,] = 1$ if $\R[j,i,] = 1$
and $P_{i+1} = T_{j+1}$, and $\R[j+1,i+1,] = 0$ otherwise.
Furthermore, $\R[j+1,1,] = 1$ if $P_{1} = T_{j+1}$ and 0 otherwise.
If $\R[j+1,p,] = 1$ then a complete match can be reported.

The transition from~$\R[j,,]$ to~$\R[j+1,,]$ can be computed very fast as follows.
For each $x \in \ALF$ let~$\D[x,]$ be a bit array of size~$p$ such that for $1 \le i \le p, \D[x,i] = 1$ if and only if $P_i = x$.
The array~$\D[x,]$ denotes the positions of the symbol~$x$ in the pattern~$P$.
Each~$\D[x,]$ can be preprocessed before the search.
The computation of~$\R[j+1,,]$ is then reduced to three bitwise operations, namely $\R[j+1,,] = (\mbox{LShift}(\R[j,,]) \OR 1) \AND \D[T_{j+1},]$.
When $\R[j,p,] = 1$, the algorithm reports a match on a position $j - p + 1$.
\section{Our Algorithm}
\label{chap:our}
In this section we will show an algorithm which solves \SM.
We call the algorithm \GSM (Graph \SM).
\GSM uses the \GTM presented in Section~\ref{sec:model} and is based on the \SHIFTAND algorithm from Section~\ref{sec:shift-and}.

The basic idea of the \GSM algorithm is to represent prefix match signals (see Definition~\ref{def:prefixmatchsignal}) from the basic matching algorithm (Section~\ref{sec:basicMatching}) over  $\G[P,P]$ in bit vectors.
The GSM algorithm represents all signals~$I$ in the bitmaps~$\RR[X,,]$ formed by three vectors, one for each row.
Each time \GSM processes a symbol of~$T$, it first propagates the signal along the edges, then filters the signal and finally checks for matches.
All these operations can be done very quickly thanks to bitwise parallelism.




First, we make the concept of \GSM more familiar by presenting a way to interpret the \SHIFTAND algorithm by means of the basic matching algorithm (\BMA) from Section~\ref{sec:basicMatching} to solve the (ordinary) Pattern Matching problem.
Then we expand this idea to \SM by using the \GTM.

\subsection{Graph Theoretic View of the Shift-And Algorithm}
\label{sec:graphshiftand}
Let~$T$ and~$P$ be a text and a pattern of lengths~$t$ and~$p$, respectively.
We create the \TGRAPH $\G[T,P] = (V,E,\LAB)$ of the pattern~$P$.

\begin{definition}
  Let~$S$ be a string. The \emph{\TGRAPH} of~$S$ is a graph $\G[T,S] = (V,E,\LAB)$ where $V = \{ v_i \mid 1 \le i \le |S| \}$, $E = \{ (v_i,v_{i+1}) \mid 1 \le i \le |S-1| \}$ and $\LAB : V \to \ALF$ such that $\LAB(v_i) = S_i$.
  \label{def:tgraph}
\end{definition}
Note that the \TGRAPH is directed acyclic graph which can be divided into columns $V_i, 1 \le i \le p$ (each of them containing one vertex~$v_i$) such that the edges lead from~$V_j$ to~$V_{j+1}$.
This means that the \TGRAPH satisfies all assumptions of \BMA.
We apply \BMA to~$\G[T,P]$ to figure out whether~$P$ matches~$T$ at a position~$j$.
We get a correct result because for each $i \in \{1,\ldots,p\}$ we check whether $T_{j+i-1} = \LAB(v_i) = P_i$.

To find every occurrence of~$P$ in~$T$ we would have to run \BMA for each position separately.
This is basically the naive approach to solve the pattern matching.
We can improve the algorithm significantly when we parallelize the computations of~$p$ runs of \BMA in the following way.

The algorithm processes one symbol at a time starting from~$T_1$.
We say that the algorithm is in the~$j\TH$ step when a symbol~$T_j$ has been processed.
\BMA represents a prefix match as a prefix match signal $\I[] : V \to \{ 0, 1 \}$.
Its value in the~$j\TH$ step is denoted~$\I[j]$.
Since one run of the \BMA uses only one column of the \TGRAPH at any time we can use other vertices to represent different runs of the \BMA.
%
We represent all prefix match indicators in one vector so that we can manipulate them easily.
To do that we prepare a bit vector~$\R[,,]$.
Its value in~$j\TH$ step is denoted~$\R[j,,]$ and defined as $\R[j,i,] = \I[j](v_i)$.

First operation which is used in \BMA (propagate signal along the edges) can be done easily by setting the signal of~$v_i$ to value of the signal of its predecessor~$v_{i-1}$ in the previous step.
I.e., for $i \in \{ 1, \dots, p \}$ we set $\I[j](v_i) = 1$ if $i=1$ and $\I[j](v_i) = \I[j-1](v_{i-1})$ otherwise.
In terms of $R^j$ this means just $R^j=\mbox{LSO}(R^{j-1})$, where \emph{LSO} is defined as $\mbox{LSO}(x) = \mbox{LShift}(x) \OR 1$.

We also need a way to set $\I[](v_i) := 0$ for each~$v_i$ for which $\LAB(v_i) \ne T_{j+i}$ which is another basic \BMA operation (filter signal by a symbol).
We can do this using the bit vector~$\D[x,]$ from Section~\ref{sec:shift-and} and taking $\R[,,] \AND \D[x,]$.
I.e., the algorithm computes $R^j$ as $\R[j,,] = \mbox{LSO}(\R[j-1,,]) \AND \D[T_{j+1},]$.

The last \BMA operation we have to define is the \emph{match detection}.
We do this by checking whether $\R[j,p,] = 1$ and if this is the case then a match starting at position $j-p+1$ occurred.


\subsection{Our Algorithm for Swap Matching Using the Graph Theoretic Model}
\label{sec:gsm}

Now we are ready to describe the GSM algorithm.

We again let $\G[P,P] = (V,E,\LAB)$ be the \PGRAPH of the pattern~$P$,
apply \BMA to~$\G[P,P]$
to figure out whether~$P$ matches~$T$ at a position~$j$, and
%
parallelize $p$ runs of \BMA on $\G[P,P]$.

Again, the algorithm processes one symbol at a time 
and it
is in the~$j\TH$ step when symbol~$T_j$ is being processed.
We again denote the value of the prefix match signal~$\I[]: V \to \{ 0, 1 \}$ of \BMA in the~$j\TH$ step by~$\I[j]$.
I.e., the semantic meaning of~$\I[j](m_{r,c})$ is that $\I[j](m_{r,c}) = 1$ if there exists a swap permutation~$\pi$ such that $\pi(c)=c+r$ and $\SS[\pi(P),1,c] = \SS[T,j-c+1,j]$. Otherwise $\I[j](m_{r,c})$ is $0$.

We want to represent all prefix match indicators in vectors so that we can manipulate them easily.
We can do this by mapping the values of $I$ for rows $r \in \{-1,0,1\}$ of the \PGRAPH to vectors $\RR[U,,],\RR[M,,]$, and~$\RR[D,,]$, respectively.
We denote value of the vector $\RR[X,,] \in \{ \RR[U,,],\RR[M,,],\RR[D,,] \}$ in~$j\TH$ step as~$\RR[X,j,]$.
We define values of the vectors as $\RR[U,j,i] = \I[j](\M[-1,i])$, $\RR[M,j,i] = \I[j](\M[0,i])$, and $\RR[D,j,i] = \I[j](\M[1,i])$, where the value of $\I[j](v)=0$ for every $v\notin V$.

We define \emph{\BMA propagate signal along the edges} operation as setting the signal of~$\M[r,c]$ to~$1$ if at least one of its predecessors have signal set to~$1$. I.e., we set
$\I[j+1](\M[-1,i]) := \I[j](\M[1,i-1])$, $\I[j+1](\M[0,i]) := \I[j](\M[-1,i-1]) \OR \I[j](\M[0,i-1])$, $\I[j+1](\M[0,1]) := 1$, $\I[j+1](\M[1,i]) := \I[j](\M[-1,i-1]) \OR \I[j](\M[0,i-1])$, and $\I[j+1](\M[1,1]) := 1$.
%
We can perform the above operation using the $\mbox{LSO}(\R[,,])$ operation.
We obtain the \emph{propagate signal along the edges} operation in the form
$\RR[U',j+1,] := \mbox{LSO}(\RR[D,j,])$, $\RR[M',j+1,] := \mbox{LSO}(\RR[M,j,] \OR \RR[U,j,])$, and
 $\RR[D',j+1,] := \mbox{LSO}(\RR[M,j,] \OR \RR[U,j,])$. 

The operation \emph{filter signal by a symbol} can be done by first constructing a bit vector~$\D[x,]$ for each $x \in \ALF$ as $\D[x,i]=1$ if $x=P_i$ and $\D[x,i]=0$ otherwise.
Then we use these vectors to filter signal by a symbol~$x$ by taking
$\RR[U,j,] := \RR[U',j,] \AND \mbox{LShift}(D^{T_j})$, $\RR[M,j,] := \RR[M',j,] \AND D^{T_j}$, and $\RR[D,j,] := \RR[D',j,] \AND \mbox{RShift}(D^{T_j})$.

The last operation we define is the match detection.
We do this by checking whether $\RR[U,j,p]=1$ or $\RR[M,j,p]=1$ and if this is the case, then a match starting at a position $j-p+1$ occurred.

%
%

\begin{algorithm}[t]
  \caption{The graph swap matching (GSM)} \label{alg:GSM}
  \begin{algorithmic}[1]
  \Statex \textbf{Input}: Pattern $P$ of length $p$ and text $T$ of length $t$ over alphabet $\ALF$.
  \Statex \textbf{Output}: Positions of all swap matches.
  \State Let $\RR[U,0,] := \RR[M,0,] := \RR[D,0,] := 0^p$. \label{step:gsm_initrs}
  \State Let $\D[x,] := 0^p$, for all $x \in \ALF$. \label{step:gsm_initd}
  \For{$i = 1,2,3,\dots,p$} \label{step:gsm_initd1}
    \State $\D[P_i, i] := 1$ \label{step:gsm_initd2}
  \EndFor
  \For{$j = 1,2,3,\dots,t$}  \label{step:gsm_repeatstep}
    \State $\RR[U',j,] := \mbox{LSO}(\RR[D,j-1,])$.\label{step:gsm_propagate1}
    \State $\RR[M',j,] := \mbox{LSO}(\RR[M,j-1,] \mid \RR[U,j-1,])$.\label{step:gsm_propagate1apul}
    \State $\RR[D',j,] := \mbox{LSO}(\RR[M,j-1,] \mid \RR[U,j-1,])$.\label{step:gsm_propagate2}%
    \State $\RR[U,j,] := \RR[U',j,] \AND \mbox{LShift}(D^{T_j})$.\label{step:gsm_filter1}%
    \State $\RR[M,j,] := \RR[M',j,] \AND D^{T_j}$.
    \State $\RR[D,j,] := \RR[D',j,] \AND \mbox{RShift}(D^{T_j})$.\label{step:gsm_filter2}%
    \If {$\RR[U,j,p] = 1$ or $\RR[M,j,p] = 1$} \label{step:gsm_check}
      \State report a match on position $j-p+1$. \label{step:gsm_end}
    \EndIf    
  \EndFor




\end{algorithmic}
\end{algorithm}

The final \GSM algorithm (Algorithm~\ref{alg:GSM}) first prepares the \DMASKS~$\D[x,]$ for every $x \in \ALF$ and initializes $\RR[U,0,]:=\RR[M,0,]:=\RR[D,0,]:=0$ (Steps~\ref{step:gsm_initrs}--\ref{step:gsm_initd2}).
Then the algorithm computes the value of vectors~$\RR[U,j,]$, $\RR[M,j,]$, and~$\RR[D,j,]$ for $j \in \{ 1,\dots,t \}$ by first using the above formula for signal propagation (Steps~\ref{step:gsm_propagate1}--\ref{step:gsm_propagate2}) and then the formula for signal filtering (Steps~\ref{step:gsm_filter1}--\ref{step:gsm_filter2}) and checks whether $\RR[U,j,p]=1$ or $\RR[M,j,p]=1$ and if this is the case the algorithm reports a match (Steps~\ref{step:gsm_check} and~\ref{step:gsm_end}).


Observe that Algorithm~\ref{alg:GSM} accesses the input sequentially and thus it is a streaming algorithm.
We now prove correctness of our algorithm.
To ease the notation let us define $R^j(m_{r,c})$ to be $\RR[U,j,c]$ if $r=-1$, $\RR[M,j,c]$ if $r=0$, and $\RR[D,j,c]$ if $r=1$. We define $R'^j(m_{r,c})$ analogously.  Similarly, we define $D^x(m_{r,c})$ as $(\mbox{LShift}(D^x))_c=D^x_{c-1}$ if $r=-1$, $D^x_c$ if $r=0$, and $(\mbox{RShift}(D^x))_c=D^x_{c+1}$ if $r=1$. By the way the masks $D^x$ are computed on lines~\ref{step:gsm_initd}--\ref{step:gsm_initd2} of Algorithm~\ref{alg:GSM}, we get the following observation.

\begin{observation}\label{obs:filter}
 For every $m_{r,i} \in V$ and every $j \in \{i, \ldots t\}$ we have $D^{T_j}(m_{r,i})=1$ if and only if $T_j=P_{r+i}$.
\end{observation}

The following lemma constitutes the crucial part of the correctness proof.

\begin{lemma}
\label{lem:match_algo}
For every $m_{r,i} \in V$ and every $j \in \{i, \ldots t\}$ we have $R^j(m_{r,i})=1$ if and only if there exists a swap permutation $\pi$ such that $\SS[\pi(P),1,i] = \SS[T,j-i+1,j]$ and $\pi(i)=i+r$.
\end{lemma}

\begin{proof}
Let us start with the ``if'' part.
We prove the claim by induction on $i$.
If $i=1$ and there is a swap permutation $\pi$ such that $\pi(1)=1+r$ and $P_{1+r}=T_{j}$, then the algorithm sets $R'^j(m_{r,1})$ to $1$ on line~\ref{step:gsm_propagate1}, \ref{step:gsm_propagate1apul}, or~\ref{step:gsm_propagate2} (recall the definition of $\mbox{LSO}$). As $P_{1+r}=T_{j}$, we have $D^{T_j}(m_{r,1})=1$ by Observation~\ref{obs:filter} and, therefore, by  lines~\ref{step:gsm_filter1}--\ref{step:gsm_filter2}, also $R^j(m_{r,1})$.

Now assume that $i>1$ and that the claim is true for every smaller $i$.
Assume that there exists a swap permutation $\pi$ such that $\SS[\pi(P),1,i] = \SS[T,j-i+1,j]$ and $\pi(i)=i+r$.
By induction hypothesis we have that $R^{j-1}(m_{r',i-1})=1$, where $r'=i-1-\pi(i-1)$. Since $r$ equals $-1$ if and only if $r'$ equals $+1$ by Definition~\ref{def:SMP}, we have $(r,r') \in \{(-1,1),(0,-1),(0,0),(1,-1),(1,0)\}$. Therefore the algorithm sets $R'^j(m_{r,i})$ to $1$ on line~\ref{step:gsm_propagate1}, \ref{step:gsm_propagate1apul}, or~\ref{step:gsm_propagate2}.
Moreover, since $P_{i+r}=T_{j}$, we have $D^{T_j}(m_{r,i})=1$ by Observation~\ref{obs:filter} and the algorithm sets $R^j(m_{r,i})$ to $1$ on one of the lines~\ref{step:gsm_filter1}--\ref{step:gsm_filter2}.

Now we prove the ``only if'' part again by induction on $i$.
If $i=1$ and $R^j(m_{r,i})=1$, then we must have $D^{T_j}(m_{r,1})=1$ and, by Observation~\ref{obs:filter}, also $P_{1+r}=T_{j}$. We obtain $\pi$ by setting $\pi(1)=1+r$, $\pi(2)=2-r$ and $\pi(i')=i'$ for every $i' \in \{2, \ldots, p\}$. It is easy to verify that this is a swap permutation for $P$ and has the desired properties.

Now assume that $i>1$ and that the claim is true for every smaller $i$.
Assume that $R^j(m_{r,i})=1$. Then, due to lines~\ref{step:gsm_filter1}--\ref{step:gsm_filter2} we must have $D^{T_j}(m_{r,i})=1$ and, hence, by Observation~\ref{obs:filter}, also $P_{i+r}=T_{j}$. Moreover, we must have $R'^j(m_{r,i})=1$ and, hence, by lines~\ref{step:gsm_propagate1}--\ref{step:gsm_propagate2} of the algorithm also $R^{j-1}(m_{r',i-1})=1$ for some $r'$ with $(r,r') \in \{(-1,1),(0,-1),(0,0),(1,-1),(1,0)\}$. By induction hypothesis there exists a swap permutation $\pi'$ for $P$ such that $\SS[\pi'(P),1,i-1] = \SS[T,j-i+1,j-1]$ and $\pi'(i-1)=i-1+r'$. If $\pi'(i)=i+r$, then setting $\pi=\pi'$ finishes the proof. Otherwise we have either $r=0$ or $r=1$ and $i < p$.
In the former case we let $\pi(i')=i'$ for every $i' \in \{i, \ldots, p\}$ and in the later case we let $\pi(i)=i+1$, $\pi(i+1)=i$ and $\pi(i')=i'$ for every $i' \in \{i+2, \ldots, p\}$. In both cases we let $\pi(i')=\pi'(i')$ for every $i' \in \{1, \ldots, i-1\}$. It is again easy to verify that $\pi$ is a swap permutation for $P$ with the desired properties.
\qed\end{proof}

\begin{theorem}
\label{thm:correct}
The \GSM algorithm is correct.
\end{theorem}

\begin{proof}

Our \GSM algorithm reports a match on position $j-p+1$ if and only if $R^j(m_{p,-1})=1$ or $R^j(m_{p,0})=1$.
However, by Lemma~\ref{lem:match_algo}, this happens if and only if there is a swap match of $P$ on position $j-p+1$ in $T$. Hence, the algorithm is correct.
\end{proof}

\begin{theorem}
\label{theor:complexity}
The \GSM algorithm runs in $O(\POverWcc (|\ALF| + t) + p)$ time and uses
$O(\POverWcc |\ALF|)$ memory cells (not counting the input and output cells),
where~$t$ is the length of the input text,~$p$ length of the input pattern,~$w$
is the word-size of the machine, and~$|\ALF|$ size of the alphabet.\footnote{To simplify the analysis, we assume that $\log t < w$, i.e., the iteration counter fits into one memory cell.}
\end{theorem}

\begin{proof}
The initialization of $\RR[X,,]$ and $\D[x,]$ masks (lines~\ref{step:gsm_initrs} and \ref{step:gsm_initd})
takes $O(\POverWcc |\ALF|)$ time.
The bits in $\D[x,]$ masks are set according to the pattern in $O(p)$ time (lines~\ref{step:gsm_initd1} and~\ref{step:gsm_initd2}).
The main cycle of the algorithm (lines~\ref{step:gsm_repeatstep}--\ref{step:gsm_end}) makes $t$ iterations.
Each iteration consists of computing values of $\RR[X,,]$ in $13$ bitwise operations, i.e., in $O(\POverWcc)$ machine operations, and checking for the result in $O(1)$ time.
This gives $O(\POverWcc (|\ALF| + t) + p)$ time in total.
The algorithm saves 3 $\RR[X,,]$ masks (using the same space for all $j$ and also for $\RR[X',,]$ masks), $|\ALF|$ $\D[x,]$ masks, and constant number of variables for other uses (iteration counters, temporary variable, etc.).
Thus, in total the \GSM algorithm needs $O(\POverWcc |\ALF|)$ memory cells.
\qed\end{proof}

\begin{corollary}
\label{cor:small_patt}
If $p = c w$ for some constant~$c$, then the GSM algorithm runs in $O(|\ALF| + p + t)$ time and has $O(|\ALF|)$ space complexity. Moreover, if $p \le w$, then the GSM algorithm can be implemented using only $7 + |\ALF|$ memory cells.
\end{corollary}

\begin{proof}
The first part follows directly from Theorem~\ref{theor:complexity}. Let us show the second part.
We need~$|\ALF|$ cells for all \DMASKS,~$3$ cells for~$\RR[,,]$ vectors (reusing the space also for $\RR[',,]$ vectors), one pointer to the text, one iteration counter, one constant for the match check and one temporary variable for the computation of the more complex parts of the algorithm. Alltogether, we need only $7 + |\ALF|$ memory cells to run the \GSM algorithm.
\qed\end{proof}

From the space complexity analysis we see that for some sufficiently small alphabets (e.g. DNA sequences) the GSM algorithm can be implemented in practice using solely CPU registers with the exception of text which has to be loaded from the RAM.

\section{Limitations of the Finite Deterministic Automata Approach}
\label{sec:limits}
Many of the string matching problems can be solved by finite automata.
The construction of a non-deterministic finite automaton that solves \SM can be done by a simple modification of the \PGRAPH.
An alternative approach to solve the \SM would thus be to determinize and execute this automaton.
The drawback is that the determinization process may lead to an exponential number of states.
We show that in some cases it actually does, contradicting the conjecture of Holub~\cite{holub}, stating that the number of states of this determinized automaton is $O(p)$.


\begin{theorem}
\label{thm:determ}
There is an infinite family~$F$ of patterns such that any deterministic finite automaton~$A_P$ accepting the language~$L_S(P)=\{u\pi(P) \mid u \in \ALF^*, \pi \text{ is a swap permutation for } P\}$ for $P \in F$ has $2^{\Omega(|P|)}$ states.
\end{theorem}

\begin{proof}
For any integer~$k$ we define the pattern $P_k := ac(abc)^k$.
Note that the length of~$P_k$ is $\Theta(k)$.
Suppose that the automaton~$A_P$ recognizing language~$L(P)$ has~$s$ states such that $s < 2^k$.
We consider a set of strings $T_0, \dots, T_{2^k-1}$ where~$T_i$ is defined as follows.
Let $b^i_{k-1}, b^i_{k-2} \dots b^i_0$ be the binary representation of the number~$i$.
Let $B^i_j=abc$ if $b^i_j = 0$ and let $B^i_j=bac$ if $b^i_j = 1$.
Then, let $T_i := ac B^i_{k-1} B^i_{k-2} \dots B^i_0$.
See Table~\ref{fig:automataTexts} for an example.
Note that each $T_i$, $i \in \{0, \ldots, 2^k-1\}$ is a swapped version of $P=T_0$. 
Since $s < 2^k$, there exist $0 \le i < j \le 2^k-1$ such that both~$T_i$ and~$T_j$ are accepted by the same accepting state~$q$ of the automaton~$A$.
Let~$m$ be the minimum number such that $b^i_{k-1-m} \ne b^j_{k-1-m}$.
Note that $b^i_m = 0$ and $b^j_m = 1$.
Now we define $T_i'=T_i (abc)^{(m+1)}$ and $T_j'=T_j (abc)^{(m+1)}$.
Let $X=(T_i')_{[3(m+1)+1,3(m+1+k)+2]}$ and $Y=(T_j')_{[3(m+1)+1,3(m+1+k)+2]}$ be the suffices of the strings $T_i'$ and $T_j'$ both of length $3k+2$.
Note that~$X$ begins with $bc\dots$ and~$Y$ begins with $ac\dots$ and that block~$abc$ or~$bac$ repeats for~$k$ times in both.
Therefore pattern~$P$ swap matches~$Y$ and does not swap match~$X$.
Since for the last symbol of both~$T_i$ and~$T_j$ the automaton is in the same state~$q$, the computation for~$T_i'$ and~$T_j'$ must end in the same state~$q'$.
However as~$X$ should not be accepted and~$Y$ should be accepted we obtain contradiction with the correctness of the automaton~$A$.
Hence, we may define the family~$F$ as $F=\{ P_1, P_2, \dots \}$, concluding the proof.
\qed\end{proof}
\begin{table}[t]
  \centering
  \caption{An example of the construction from proof of Theorem~\ref{thm:determ} for~$k=3$.}
  \begin{tabular}{r|c}
   $P=T_0$ & $acabcabcabc$ \\
      $T_1$ & $acabcabcbac$ \\
      $T_2$ & $acabcbacabc$ \\
      $T_3$ & $acabcbacbac$ \\
      $T_4$ & $acbacabcabc$ \\
      $T_5$ & $acbacabcbac$ \\
      $T_6$ & $acbacbacabc$ \\
      $T_7$ & $acbacbacbac$ \\
  \end{tabular}
  \label{fig:automataTexts}
\end{table}
This proof shows the necessity for specially designed algorithms which solve the \SM.
We presented one in the previous section and now we reiterate on the existing algorithms.

\section{Smalgo Algorithm}
\label{chap:smalgo}
In this section we discuss how \SMO by Iliopoulos and Rahman~\cite{newModel} and \SMI and \SMII by Ahmed et al.~\cite{revisited} work.
Since \SMI is bitwise inverse of \SMO, we will introduce them both in terms of operations used in \SMI.
After that we will describe and analyze \SMII.

Before we show how these algorithms work, we need one more definition.
\begin{definition}
  \label{def:degenerate}
  A \emph{degenerate symbol}~$w$ over an alphabet~$\ALF$ is a nonempty set of symbols from alphabet~$\ALF$.
  A \emph{degenerate string}~$S$ is a string built over an alphabet of degenerate symbols.
We say that a degenerate string~$\widetilde P$ matches a text~$T$ at a position~$j$ if $T_{j+i-1} \in \widetilde{P_i}$ for every $1 \le i \le p$.
\end{definition}

\subsection{\SMI}
The \SMI \cite{revisited} algorithm is a modification of the \SHIFTAND algorithm from Section~\ref{sec:shift-and} for \SM.
The algorithm uses the \GTM introduced in Section~\ref{sec:model}.

First let $\widetilde P = \{P_{1}, P_{2} \} \dots \{P_{x-1}, P_{x}, P_{x+1} \} \dots \{ P_{p-1}, P_{p} \}$ be a a degenerate version of pattern~$P$.
The symbol on position~$i$ in~$\widetilde P$ represents the set of symbols of~$P$ which can swap to that position.
To accommodate the \SHIFTAND algorithm to match degenerate patterns we need to change the way the~$\D[x,]$ masks are defined.
For each $x \in \ALF$ let~$\widetilde D^x_i$ be the bit array of size~$p$ such that for $1 \le i \le p, \widetilde D^x = 1$ if and only if $x \in \widetilde P_i$.

While a match of the degenerate pattern~$\widetilde P$ is a necessary condition for a swap match of~$P$, it is clearly not sufficient.
The way the \SMALGO algorithms try to fix this is by introducing \emph{\PMASK} $P(x_1,x_2,x_3)$ which is defined as
$P(x_1,x_2,x_3)_i = 1$ if  $i = 1$ or if there exist vertices $u_1,u_2$, and~$u_3$ and edges $(u_1,u_2), (u_2,u_3)$ in~$\G[P,P]$ for which $u_2 = \M[r,i]$ for some $r \in \{-1,0,1\}$ and $\LAB(u_n) = x_n$ for $1 \le n \le 3$, and $P(x_1,x_2,x_3)_i = 0$ otherwise.
One~$P$-mask called $P(x,x,x)$ is used to represent the~$P$-masks for triples $(x_1,x_2,x_3)$ which only contain 1 in the first column.

Now, whenever checking whether~$P$ prefix swap matches~$T$~$k+1$ symbols at position~$j$ we check for a match of~$\widetilde P$ in~$T$ and we also check whether $P(T_{j+k-1},T_{j+k},T_{j+k+1})_{k+1} = 1$.
This ensures that the symbols are able to swap to respective positions and that those three symbols of the text~$T$ are present in some~$\pi(P)$.

%

With the \PMASKS completed we initialize $\R[1,,]=1 \AND \widetilde D^{T_1}$. Then for every~$j=1$ to~$t$ we repeat the following. We compute~$\R[j+1,,]$ as $\R[j+1,,] = \mbox{LSO}(\R[j,,]) \AND \widetilde D^{T_{j+1}} \AND \mbox{RShift}(\widetilde D^{T_{j+2}}) \AND P(T_j, T_{j+1}, T_{j+2})$.
To check whether or not a swap match occurred we check whether $\R[j,p-1,] = 1$.
This is claimed to be sufficient because during the processing we are in fact considering not only the next symbol~$T_{j+1}$ but also the symbol~$T_{j+2}$.

\subsection{The Flaw in the \SMO, \SMI and \SMII}
We shall see that for a pattern $P = abab$ and a text $T = aaba$ all \SMO versions give false positives.

The concept of \SMO is based on the assumption that we can find a path in~$\G[P,P]$ by searching for consecutive paths of length~$3$ (\emph{triplets}), where each two consecutive share two columns and can partially overlap.
However, this only works if the consecutive triplets \emph{actually} share the two vertices in the common columns.
If the assumption is not true then
the found substring of the text might not match any swapped version of~$P$.

The above input gives such a configuration (see Fig.~\ref{fig:triplets}) and therefore the assumption is false.
The \SMI algorithm actually reports match of pattern $P = abab$ on a position~$1$ of text $T = aaba$.
This is obviously a false positive, as the pattern has two~$b$ symbols while the text has only one.

The reason behind the false positive match is as follows.
The algorithm checks whether the first triplet of symbols $(a,a,b)$ matches.
It can match the swap pattern~$aabb$.
Next it checks the second triplet of symbols $(a,b,a)$, which can match~$baba$.
We know that~$baba$ is not possible since it did not appear in the previous check, but the algorithm cannot distinguish them since it only checks for triplets existence.
Since each step gave us a positive match the algorithm reports a swap match of the pattern in the text.

In the Fig.~\ref{fig:triplets} we see the two triplets which \SMO assumes have two vertices in common.
The \SMII algorithm saves space by maintaining less information, however it simulates how \SMI works and so it contains the same flaw.

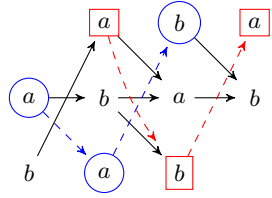
\begin{figure}[t]
  \centering
  \begin{tikzpicture}
    \node[blue]      (11)                 {$a$};
    \node[main]      (21) [right of = 11] {$b$};
    \node[main]      (31) [right of = 21] {$a$};
    \node[main]      (41) [right of = 31] {$b$};

    \node[red]       (20) [above of = 21] {$a$};
    \node[blue]      (30) [above of = 31] {$b$};
    \node[red]       (40) [above of = 41] {$a$};

    \node[main]      (12) [below of = 11] {$b$};
    \node[blue]      (22) [below of = 21] {$a$};
    \node[red]       (32) [below of = 31] {$b$};

    \path[every node/.style={font=\sffamily\small}]
    (11) edge (21)
    (21) edge (31)
    (31) edge (41)

    (11) edge [blue,dashed] (22)
    (21) edge (32)

    (12) edge (20)
    (22) edge [blue,dashed] (30)
    (32) edge [red,dashed] (40)

    (20) edge (31)
    (30) edge (41)

    (20) edge [shorten >= \ShortArrow, red,  dashed, bend right = 10] (32)
    ;
  \end{tikzpicture}
  \caption{\SMO flaw represented in the \PGRAPH for $P = abab$}
  \label{fig:triplets}
\end{figure}

\subsection{The Run of \SMI Resulting in the False Positive}
In Tables~\ref{tab:flawmasks} and~\ref{tab:flawrun} we can see the step by step execution of \SMI algorithm on pattern $P = abab$ and text $T = aaba$.
In Table~\ref{tab:flawrun} we see that~$\R[3,,]$ has~$1$ in the~$3\RD$ row which means that the algorithm reports a pattern match on a position~$1$.
This is a false positive, because it is not possible to swap match the pattern with two~$b$ symbols in the text with only one~$b$ symbol.
\begin{table}[t]
    \centering
    \caption{$\widetilde D$-masks and \PMASKS for $P = abab$. A column $xyz$ contains values $P(x,y,z)_i$.}
    \label{tab:flawmasks}
    \tabcolsep=0.11cm
    \begin{tabular}{|c|c|c|c|c|c|c|c|c|c|c|c|}
      \hline%
      \raisebox{0pt}[1.2em]{}i &  $\widetilde P_i$ &$\widetilde D^a_i$&$\widetilde D^b_i$&
      $aaa$
      &$aab$&$aba$&$baa$
      &$abb$&$bab$&$bba$
      &$bbb$\\
      \hline
      1 & $[ab]$  & 1 & 1 & 1 & 1 & 1 & 1 & 1 & 1 & 1 & 1 \\
      2 & $[ba]$  & 1 & 1 & 0 & 1 & 1 & 1 & 1 & 1 & 0 & 0 \\
      3 & $[ab]$  & 1 & 1 & 0 & 1 & 1 & 0 & 1 & 1 & 1 & 0 \\
      4 & $[ba]$  & 1 & 1 & 0 & 0 & 0 & 0 & 0 & 0 & 0 & 0 \\
      \hline
    \end{tabular}
\end{table}
\begin{table}[t]
  \centering
  \caption{\SMI algorithm execution for $P = abab$ and $T = aaba$. The column $RD^x$ denotes the values of $\mbox{RShift}(\widetilde D^x)$.}
   \label{tab:flawrun}
  \tabcolsep=0.11cm
  \begin{tabular}{|c|c|c|c|c|c|c|c|c|c|c|c|}
    \hline%
    \raisebox{0pt}[1.2em][0pt]{}i &$\R[1,,]$&$\mbox{LSO}(\R[1,,])$&$\widetilde D^{a}$&$RD^{b}$&$P(a,a,b)$
    &$\R[2,,]$&$\mbox{LSO}(\R[2,,])$&$\widetilde D^{b}$&$RD^{a}$&$P(a,b,a)$
    &$\R[3,,]$\\
    \hline
    1 & 1 & 1 & 1 & 1 & 1  & 1 & 1 & 1 & 1 & 1  & 1 \\
    2 & 0 & 1 & 1 & 1 & 1  & 1 & 1 & 1 & 1 & 1  & 1 \\
    3 & 0 & 0 & 1 & 1 & 1  & 0 & 1 & 1 & 1 & 1  & 1 \\
    4 & 0 & 0 & 1 & 0 & 0  & 0 & 0 & 1 & 0 & 0  & 0 \\
    \hline
  \end{tabular}
\end{table}

\subsection{Description of \SMII}
{
\newcommand{\pmask}{\text{pmask}}
\newcommand{\up}{\text{up}}
\newcommand{\down}{\text{down}}
\newcommand{\midl}{\text{middle}}
\newcommand{\checkup}{\text{checkUp}}
\newcommand{\checkdown}{\text{checkDown}}
\newcommand{\prevcheckup}{\text{prevCheckUp}}
\newcommand{\prevcheckdown}{\text{prevCheckDown}}
\newcommand{\temp}{\text{temp}}
\renewcommand{\neg}{\ \sim\mkern-3mu}
\newcommand{\patternlength}{\text{patternLength}}

To explain the \SMII algorithm in more detail, we first introduce a notion of \emph{change}. An \emph{upward change} corresponds to (the \BMA) going to vertex $\M[-1,i]$ for some~$i$, a \emph{downward change} corresponds to going to vertex $\M[+1,i]$, and a \emph{middle-ward change} corresponds to going to vertex $\M[0,i]$.

If a \emph{downward change} has occurred, then we have to check whether an \emph{upward change} occurs at the next position. If an \emph{upward change} has occurred, then we have to check whether a \emph{downward} or \emph{middle-ward change} occurs at the next position. The main problem here is how to tell whether the changes \emph{actually} occur.

To this end, the authors of the algorithm introduce three new types of masks, namely \emph{up-masks} $\up_{(x,y)}$, \emph{down-masks} $\down_{(x,y)}$, and \emph{middle-masks} $\midl_{(x,y)}$, which express whether an \emph{upward}, a \emph{downward}, and a \emph{middle-ward change} can occur at the particular position, respectively, with the endpoints of the edge having labels $x$ and $y$.

The authors of the algorithm now claim that to perform the above checks, it is enough to save the previous \emph{down-mask} and match its value with current \emph{up-mask} and $R_j$, or to save the previous \emph{up-mask} and match its value with current \emph{down-mask}, \emph{middle-mask}, and $R_j$, respectively. However, this way in both cases we only check whether the change can occur, not whether it actually occurred. This would lead not only to false positives (as shown in Section~\ref{chap:smalgo}), but also to false negatives.

Unfortunately, no more details are available about the algorithm in the original paper.
The pseudocode of \SMII (which contains numerous errors) performs something different and we include its analysis in the next section 
for completeness.
Nevertheless, the example presented in the Section~\ref{chap:smalgo} and in the previous section still makes the pseudocode (with the small errors corrected) report a false positive.

\subsection{Analysis of the Pseudocode of \SMII}

\label{app:analysis}
In this section we analyze the pseudocode of the \SMII algorithm as given by Ahmed et al. in~\cite{revisited}, we will perform equivalent transformation on the pseudocode in order to understand the meaning of the checks the pseudocode actually performs.

The original pseudocode is as follows.
\begin{algorithm}[H]
\caption{\SMII}
\label{Algorithm for Approximate String Matching Allowing for Fixed Length Translocation(Improved)}
\begin{algorithmic}[1]
\Statex \textbf{Require:} \textbf{Text T}, \textbf{up-mask up}, \textbf{down-mask down}, \textbf{middle-mask middle}, \textbf{P-mask pmask},   \textbf{D-mask D}  for given pattern p
\State{$R_{0} \leftarrow 2^{\patternlength-1}$}
\State{$\checkup \leftarrow \checkdown \leftarrow 0$}
\State{$R_{0} \leftarrow R_{0}$ \& $D_{T_{0}}$}
\State{$R_{1} \leftarrow R_{0} \gg 1$}
\For {$j=0$ to $(n-2)$}
	\State{$R_{j} \leftarrow R_{j}$ \& $pmask_{(T_{j},T_{j+1})}$ \& $D_{T_{j+1}}$}
	\State{$\temp \leftarrow \prevcheckup \gg 1$}
	\State{$\checkup \leftarrow \checkup$ $|$  $\up_{(T_{j},T_{j+1})}$}
	\State{$\checkup \leftarrow \checkup$ \& $\neg \down_{(T_{j},T_{j+1})}$ \& $\neg \midl_{(T_{j},T_{j+1})}$}
	\State{$\prevcheckup \leftarrow \checkup$}
	\State{$R_{j} \leftarrow \neg (\temp$ \& $\checkup)$ \& $R_{j}$}
	\State{$\temp \leftarrow \prevcheckdown \gg 1$}
	\State{$\checkdown \leftarrow \checkdown$ $|$  $\down_{(T_{j},T_{j+1})}$}
	\State{$\checkdown \leftarrow \checkdown$ \& $\neg \up_{(T_{j},T_{j+1})}$}
	\State{$\prevcheckdown \leftarrow \checkdown$}
	\State{$R_{j} \leftarrow \neg (\temp$ \& $\checkdown)$ \& $R_{j}$}
	\If{$(R_{j}$ \& $1) = 1$} \State{Match found ending at position $(j-1)$}
	\EndIf
	\State{$R_{j+1} \leftarrow R_{j} \gg 1$}
	\State{$\checkup \leftarrow \checkup \gg 1$}
	\State{$\checkdown \leftarrow \checkdown \gg 1$}
\EndFor
\end{algorithmic}
\end{algorithm}
The pseudocode has several problems. First, in the first iteration of the cycle, the algorithm uses the value of the variable $\prevcheckup$ which was never initialized. Second, the algorithm never adds new ones to the variable $R$ and, hence, can never report a match after position $\patternlength$ of the text. Third, if the text is of the same length as the pattern, the algorithm only applies the shift $\patternlength-2$ times to the original value of
$2^{\patternlength-1}$ (note that in the first iteration it uses $R_0$ and overwrites the value of $R_1$) before the last match check. Therefore, at the last check, the value could only drop to $2^{\patternlength-1-\patternlength+2}=2^1=2$ and the match check cannot be successful. Also the reported position of the match does not make much sense.

Let us first correct all these easy problems.
\begin{algorithm}[H]
\caption{\SMII}
\begin{algorithmic}[1]
\State{$R_{0} \leftarrow 2^{\patternlength-1}$}
\State{$\prevcheckup \leftarrow \prevcheckdown \leftarrow \checkup \leftarrow \checkdown \leftarrow 0$}
\State{$R_{0} \leftarrow R_{0}$ \& $D_{T_{0}}$}
\State{$R_{1} \leftarrow R_{0} \gg 1$}
\For {$j=0$ to $(n-2)$}
	\State{$R_{j+1} \leftarrow R_{j+1}$ \& $pmask_{(T_{j},T_{j+1})}$ \& $D_{T_{j+1}}$}
	\State{$\temp \leftarrow \prevcheckup \gg 1$}
	\State{$\checkup \leftarrow \checkup$ $|$  $\up_{(T_{j},T_{j+1})}$}
	\State{$\checkup \leftarrow \checkup$ \& $\neg \down_{(T_{j},T_{j+1})}$ \& $\neg \midl_{(T_{j},T_{j+1})}$}
	\State{$\prevcheckup \leftarrow \checkup$}
	\State{$R_{j+1} \leftarrow \neg (\temp$ \& $\checkup)$ \& $R_{j+1}$}
	\State{$\temp \leftarrow \prevcheckdown \gg 1$}
	\State{$\checkdown \leftarrow \checkdown$ $|$  $\down_{(T_{j},T_{j+1})}$}
	\State{$\checkdown \leftarrow \checkdown$ \& $\neg \up_{(T_{j},T_{j+1})}$}
	\State{$\prevcheckdown \leftarrow \checkdown$}
	\State{$R_{j+1} \leftarrow \neg (\temp$ \& $\checkdown)$ \& $R_{j+1}$}
	\If {$(R_{j+1}$ \& $1) = 1$} \State{Match found ending at position $(j+1)$}
	\EndIf
	\State{$R_{j+2} \leftarrow (R_{j+1} \gg 1) \mid 2^{\patternlength-1}$}
	\State{$\checkup \leftarrow \checkup \gg 1$}
	\State{$\checkdown \leftarrow \checkdown \gg 1$}
\EndFor
\end{algorithmic}
\end{algorithm}

If we now move the line setting $\prevcheckup$ to $\checkup$ after the line where the check with the $\temp$ variable is performed and similarly with $\prevcheckdown$, we do not need the $\temp$ variable anymore. We also move the shifts of $\checkup$ and $\checkdown$ closer to where this variables are used. We only show the important part of the algorithm.

\begin{algorithm}[H]
\caption{\SMII}
\begin{algorithmic}[1]
\setcounter{ALG@line}{4}
\Statex \dots
\For {$j=0$ to $(n-2)$}
	\State{$R_{j+1} \leftarrow R_{j+1}$ \& $pmask_{(T_{j},T_{j+1})}$ \& $D_{T_{j+1}}$}
	\State{$\checkup \leftarrow \checkup$ $|$  $\up_{(T_{j},T_{j+1})}$}
	\State{$\checkup \leftarrow \checkup$ \& $\neg \down_{(T_{j},T_{j+1})}$ \& $\neg \midl_{(T_{j},T_{j+1})}$}
	\State{$R_{j+1} \leftarrow \neg (\prevcheckup \gg 1$ \& $\checkup)$ \& $R_{j+1}$}
	\State{$\prevcheckup \leftarrow \checkup$}
	\State{$\checkup \leftarrow \checkup \gg 1$}
	\State{$\checkdown \leftarrow \checkdown$ $|$  $\down_{(T_{j},T_{j+1})}$}
	\State{$\checkdown \leftarrow \checkdown$ \& $\neg \up_{(T_{j},T_{j+1})}$}
	\State{$R_{j+1} \leftarrow \neg (\prevcheckdown \gg 1$ \& $\checkdown)$ \& $R_{j+1}$}
	\State{$\prevcheckdown \leftarrow \checkdown$}
	\State{$\checkdown \leftarrow \checkdown \gg 1$}
	\If {$(R_{j+1}$ \& $1) = 1$} \State{Match found ending at position $(j+1)$}
	\EndIf
	\State{$R_{j+2} \leftarrow (R_{j+1} \gg 1) \mid 2^{\patternlength-1}$}
\EndFor
\end{algorithmic}
\end{algorithm}

Now we swap the order of setting $\prevcheckup$ to $\checkup$ and the shift of $\checkup$. As this makes $\prevcheckup$ shifted by one, we remove the additional shift in the check. Similarly for $\checkdown$.

\begin{algorithm}[H]
\caption{\SMII}
\begin{algorithmic}[1]
\setcounter{ALG@line}{6}
\Statex \dots
	\State{$\checkup \leftarrow \checkup$ $|$  $\up_{(T_{j},T_{j+1})}$}
	\State{$\checkup \leftarrow \checkup$ \& $\neg \down_{(T_{j},T_{j+1})}$ \& $\neg \midl_{(T_{j},T_{j+1})}$}
	\State{$R_{j+1} \leftarrow \neg (\prevcheckup$ \& $\checkup)$ \& $R_{j+1}$}
	\State{$\checkup \leftarrow \checkup \gg 1$}
	\State{$\prevcheckup \leftarrow \checkup$}
	\State{$\checkdown \leftarrow \checkdown$ $|$  $\down_{(T_{j},T_{j+1})}$}
	\State{$\checkdown \leftarrow \checkdown$ \& $\neg \up_{(T_{j},T_{j+1})}$}
	\State{$R_{j+1} \leftarrow \neg (\prevcheckdown$ \& $\checkdown)$ \& $R_{j+1}$}
	\State{$\checkdown \leftarrow \checkdown \gg 1$}
	\State{$\prevcheckdown \leftarrow \checkdown$}
\Statex \dots
\end{algorithmic}
\end{algorithm}

Now we institute $\checkup$ into the check and move its computation after the check.

\begin{algorithm}[H]
\caption{\SMII}
\begin{algorithmic}[1]
\setcounter{ALG@line}{5}
\Statex \dots
	\State{$R_{j+1} \leftarrow R_{j+1}$ \& $pmask_{(T_{j},T_{j+1})}$ \& $D_{T_{j+1}}$}
	\State{$R_{j+1} \leftarrow \neg (\prevcheckup$ \& ($\checkup \mid \up_{(T_{j},T_{j+1})})$  \& $\neg \down_{(T_{j},T_{j+1})}$ \& $\neg \midl_{(T_{j},T_{j+1})})$ \& $R_{j+1}$}
	\State{$\checkup \leftarrow (\checkup \mid \up_{(T_{j},T_{j+1})})$ \& $\neg \down_{(T_{j},T_{j+1})}$ \& $\neg \midl_{(T_{j},T_{j+1})}$}
	\State{$\checkup \leftarrow \checkup \gg 1$}
	\State{$\prevcheckup \leftarrow \checkup$}
	\State{$R_{j+1} \leftarrow \neg (\prevcheckdown$ \& $(\checkdown \mid \down_{(T_{j},T_{j+1})})$ \& $\neg \up_{(T_{j},T_{j+1})})$ \& $R_{j+1}$}
	\State{$\checkdown \leftarrow (\checkdown \mid \down_{(T_{j},T_{j+1})})$ \& $\neg \up_{(T_{j},T_{j+1})}$}
	\State{$\checkdown \leftarrow \checkdown \gg 1$}
	\State{$\prevcheckdown \leftarrow \checkdown$}
\Statex \dots
\end{algorithmic}
\end{algorithm}

Now note that during the check, the content of $\prevcheckup$ is exactly the same as the content of $\checkup$, so we can remove $\prevcheckup$ completely.

\begin{algorithm}[H]
\caption{\SMII}
\begin{algorithmic}[1]
\State{$R_{0} \leftarrow 2^{\patternlength-1}$}
\State{$\checkup \leftarrow \checkdown \leftarrow 0$}
\State{$R_{0} \leftarrow R_{0}$ \& $D_{T_{0}}$}
\State{$R_{1} \leftarrow R_{0} \gg 1$}
\For {$j=0$ to $(n-2)$}
	\State{$R_{j+1} \leftarrow R_{j+1}$ \& $pmask_{(T_{j},T_{j+1})}$ \& $D_{T_{j+1}}$}
	\State{$R_{j+1} \leftarrow \neg (\checkup$ \& ($\checkup \mid \up_{(T_{j},T_{j+1})})$  \& $\neg \down_{(T_{j},T_{j+1})}$ \& $\neg \midl_{(T_{j},T_{j+1})})$ \& $R_{j+1}$}
	\State{$\checkup \leftarrow (\checkup \mid \up_{(T_{j},T_{j+1})})$  \& $\neg \down_{(T_{j},T_{j+1})}$ \& $\neg \midl_{(T_{j},T_{j+1})}$}
	\State{$\checkup \leftarrow \checkup \gg 1$}
	\State{$R_{j+1} \leftarrow \neg (\checkdown$ \& $(\checkdown \mid \down_{(T_{j},T_{j+1})})$ \& $\neg \up_{(T_{j},T_{j+1})})$ \& $R_{j+1}$}
	\State{$\checkdown \leftarrow (\checkdown \mid \down_{(T_{j},T_{j+1})})$ \& $\neg \up_{(T_{j},T_{j+1})}$}
	\State{$\checkdown \leftarrow \checkdown \gg 1$}
	\If {$(R_{j+1}$ \& $1) = 1$} \State{Match found ending at position $(j+1)$}
	\EndIf
	\State{$R_{j+2} \leftarrow (R_{j+1} \gg 1) \mid 2^{\patternlength-1}$}
	\EndFor
\end{algorithmic}
\end{algorithm}

Now we modify the expressions by laws of logic to arrive at the following formulation.

\begin{algorithm}[H]
\caption{\SMII}
\begin{algorithmic}[1]
\setcounter{ALG@line}{6}
\Statex \dots
	\State{$R_{j+1} \leftarrow R_{j+1}$ \& $(\neg \checkup \mid  \down_{(T_{j},T_{j+1})} \mid \midl_{(T_{j},T_{j+1})})$}
	\State{$\checkup \leftarrow (\checkup$ \& $\neg \down_{(T_{j},T_{j+1})}$ \& $\neg \midl_{(T_{j},T_{j+1})}) \mid (\up_{(T_{j},T_{j+1})}$ \& $\neg \down_{(T_{j},T_{j+1})}$ \& $\neg \midl_{(T_{j},T_{j+1})})$}
	\State{$\checkup \leftarrow \checkup \gg 1$}
	\State{$R_{j+1} \leftarrow  R_{j+1}$ \& $(\neg \checkdown \mid \up_{(T_{j},T_{j+1})})$}
	\State{$\checkdown \leftarrow (\checkdown$ \& $\neg \up_{(T_{j},T_{j+1})})\mid (\down_{(T_{j},T_{j+1})}$ \& $\neg \up_{(T_{j},T_{j+1})})$}
	\State{$\checkdown \leftarrow \checkdown \gg 1$}
\Statex \dots
\end{algorithmic}
\end{algorithm}

Now, if the first subexpression in the logical OR setting the new value of $\checkup$ is true, then the appropriate bit of $R_{j+1}$ was just set to $0$ on the previous line and filtrating this bit again in future is useless. Hence, we can omit this part of the expression. We arrive at the following resulting pseudocode.

\begin{algorithm}[H]
\caption{\SMII}
\begin{algorithmic}[1]
\State{$R_{0} \leftarrow 2^{\patternlength-1}$}
\State{$\checkup \leftarrow \checkdown \leftarrow 0$}
\State{$R_{0} \leftarrow R_{0}$ \& $D_{T_{0}}$}
\State{$R_{1} \leftarrow R_{0} \gg 1$}
\For {$j=0$ to $(n-2)$}
	\State{$R_{j+1} \leftarrow R_{j+1}$ \& $pmask_{(T_{j},T_{j+1})}$ \& $D_{T_{j+1}}$}
	\State{$R_{j+1} \leftarrow R_{j+1}$ \& $(\neg \checkup \mid  \down_{(T_{j},T_{j+1})} \mid \midl_{(T_{j},T_{j+1})})$}
	\State{$\checkup \leftarrow \up_{(T_{j},T_{j+1})}$ \& $\neg \down_{(T_{j},T_{j+1})}$ \& $\neg \midl_{(T_{j},T_{j+1})}$}
	\State{$\checkup \leftarrow \checkup \gg 1$}
	\State{$R_{j+1} \leftarrow  R_{j+1}$ \& $(\neg \checkdown \mid \up_{(T_{j},T_{j+1})})$}
	\State{$\checkdown \leftarrow \down_{(T_{j},T_{j+1})}$ \& $\neg \up_{(T_{j},T_{j+1})}$}
	\State{$\checkdown \leftarrow \checkdown \gg 1$}
	\If {$(R_{j+1}$ \& $1) = 1$} \State{Match found ending at position $(j+1)$}
	\EndIf
	\State{$R_{j+2} \leftarrow (R_{j+1} \gg 1) \mid 2^{\patternlength-1}$}
	\EndFor
\end{algorithmic}
\end{algorithm}

Now it is easy to see, that $\checkup$ stores the information on whether an \textit{upward-change} must have occurred in the previous step (provided that there was a prefix match) and this is compared with the information whether \textit{downward-change} or \textit{middle-change} can occur. Similarly for the \textit{downward-change}. This is not sufficient to avoid false positives since sometimes both \textit{upward-change} and \textit{downward-change} can occur (e.g, as in our counterexample), in which case no filtration is performed at all.
}

\subsection{Why the Flaw is Not Easily Repairable}

Consider the following attempt to fix the \SMI or \SMII.
After each reported match we check for the validity of the result using a single linear-time algorithm.
This approach would rule out false positives but it ruins the time complexity of the algorithms, since there are texts of arbitrary length~$t$ with $\Theta(t)$ of reported occurrences.

Namely consider the text $T=aa(baa)^n$ for some positive~$n$, pattern $P = abab$, and let $t=|T|$.
Note that $n = (t-2)/3 = \Theta(t)$.
Text~$T$ contains string~$aaba$ on positions $1,4,7,\dots,3(n-1)+1$ ($n$ occurrences in total) and string~$baab$ on positions $3,6,\dots,3(n-1)$ ($n-1$ occurrences in total).
Thus there are~$2n-1$ occurrences which need to be checked since the~$n$ occurrences of~$aaba$ are reported by the algorithms although they are not valid matches.
Even if the checking for correctness was done in linear time~$O(p)$, the algorithms will report up to $\Theta(t)$ occurrences which means we have to run the checking algorithm $\Theta(t)$ times.
Therefore the time complexity of a version of the \SMALGO algorithm corrected this way is $\O(t p)$ even for a pattern length similar to the word-size of the target machine.

Also, the flaw cannot be resolved by checking for subpaths of length 4 or any larger constant, due to the following. Consider a pattern $P=(ab)^n$ and a text $T=aa(ba)^{n-1}$ for any positive~$n$. Obviously~$P$ does not swap match~$T$, as they are of the same length~$2n$, but~$T$ contains more~$a$'s than~$P$. However, there is a swap permutation~$\pi$ for~$P$ such that $(\pi(P))_{[1\ldots (2n-1)]} = T_{[1\ldots (2n-1)]}$ and also a swap permutation~$\pi'$ for~$P$ such that $(\pi'(P))_{[2\ldots (2n)]} = T_{[2\ldots (2n)]}$.
For example if we have $P = abab$ and a text $T = aabaabaabaa$ both \SMALGO algorithms report swap matches on positions $\{ 1,3,4,6,7 \}$ while the correct output would be $\{ 3,6 \}$.

\section{Experiments}
\label{chap:experiments}
{
\def\SMAL{SMALGO\xspace}
\def\GSMA{GSM\xspace}
\def\BPCS{BPCS\xspace}
\def\BPBCS{BPBCS\xspace}

We implemented our Algorithm~\ref{alg:GSM} (GSM), described in Section \ref{sec:gsm}, the Bitwise Parallel Cross Sampling (BPCS) algorithm by Cantone and Faro~\cite{CS}, the Bitwise Parallel Backward Cross Sampling (BPBCS) algorithm by Campanelli et al.~\cite{BCS},
and the faulty \SMAL algorithm by Iliopoulos and Rahman~\cite{newModel}.
All these implementations are available online.\footnote{\url{http://users.fit.cvut.cz/blazeva1/gsm.html}}

We tested the implementations on three real-world datasets.
The first dataset (CH) is the 7\textsuperscript{th} chromosome of the human genome\footnote{\url{ftp://ftp.ensembl.org/pub/release-90/fasta/homo_sapiens/dna/}} which consists of 159\,M characters from the standard \texttt{ACTG} nucleobases and \texttt{N} as for non-determined.
Second dataset (HS) is a partial genome of Homo sapiens from the Protein Corpus\footnote{\url{http://www.data-compression.info/Corpora/ProteinCorpus/}} with 3.3\,M characters representing proteins encoded in 19 different symbols.
The last dataset (BIB) is the Bible text of the Cantenbury Corpus\footnote{\url{http://corpus.canterbury.ac.nz/descriptions/large/bible.html}} with 4.0\,M characters containing 62 different symbols.
For each length from $3, 4, 5, 6, 8, 9, 10, 12, 16$, and $32$ we randomly selected 10,000 patterns from each text and processed each of them with each implemented algorithm.

All measured algorithms were implemented in {\tt C++} and compiled with {\tt -O3} in {\tt gcc 6.3.0}.
Measurements were carried on an Intel Core i7-4700HQ processor with 2.4\,GHz base frequency and 3.4\,GHz turbo with 8\,GiB of DDR3 memory at 1.6\,GHz.
Time was measured using {\tt std::chrono::high\_resolution\_clock::now()} from the {\tt C++ chrono} library.
The resulting running times, shown in Table~\ref{tab:dataExperiments}, were averaged over the 10,000 patterns of the given length. 
\renewcommand{\arraystretch}{1.1}
\begin{table}[t]
    \centering
    \caption{Comparison of the running times in milliseconds. Each value is the average over 10,000 patterns randomly selected from the text.}
    \setlength{\tabcolsep}{4pt}
    \begin{tabular}{|c|c|c|c|c|c|c|c|c|c|c|c|}
        \hline
        \multirow{2}{6mm}{Data ($|\Sigma|$)}&\multirow{2}{*}{Algor.}&\multicolumn{10}{|c|}{Pattern Length}\\
        \cline{3-12}
         && 3& 4& 5& 6& 8& 9& 10& 12& 16& 32 \\
        \hline\hline        
        \multirow{4}{6mm}{CH (5)}
        &\SMAL  & 426         & 376         & 355         & 350         & 347         & 347         & 344         & 347         & 345         & 345         \\
        &\BPCS  & 398         & \textbf{353}& \textbf{335}& \textbf{332}& \textbf{329}& 329         & 326         & 328         & 329         & 327         \\
        &\BPBCS & 824         & 675         & 555         & 472         & 366         & \textbf{328}& \textbf{297}& \textbf{257}& \textbf{199}& \textbf{112}\\
        &\GSMA  & \textbf{394}& 354         & 338         & 333         & 332         & 331         & 329         & 333         & 331         & 333         \\
        \hline
%
        \hline
        \multirow{4}{6mm}{HS (19)}
        &\SMAL & 4.80         & 4.73         & 4.72         & 4.74         & 4.70         & 4.71         & 4.71         & 4.71         & 4.72         & 4.70          \\
        &\BPCS & 4.43         & \textbf{4.36}& \textbf{4.36}& 4.36         & 4.34         & 4.33         & 4.34         & 4.34         & 4.35         & 4.34          \\
        &\BPBCS& 7.16         & 5.80         & 4.79         & \textbf{4.05}& \textbf{3.03}& \textbf{2.70}& \textbf{2.44}& \textbf{2.06}& \textbf{1.62}& \textbf{0.95} \\
        &\GSMA & \textbf{4.42}& 4.38         & 4.41         & 4.46         & 4.45         & 4.45         & 4.45         & 4.44         & 4.53         & 4.48          \\
        \hline
        \hline
        \multirow{4}{6mm}{BIB (62)}
        &\SMAL &  8.60        & 8.38         & 8.29         & 8.34         & 8.32         & 8.33         & 8.30         & 8.35         & 8.35         & 8.33 \\
        &\BPCS &  7.53        & \textbf{7.36}& \textbf{7.28}& 7.29         & 7.26         & 7.27         & 7.26         & 7.28         & 7.29         & 7.25 \\
        &\BPBCS& 12.43        & 10.03        & 8.26         & \textbf{7.03}& \textbf{5.44}& \textbf{4.93}& \textbf{4.52}& \textbf{3.93}& \textbf{3.19}& \textbf{1.88} \\
        &\GSMA & \textbf{7.52}&  7.37        & 7.31         & 7.35         & 7.38         & 7.40         & 7.38         & 7.42         & 7.44         & 7.40 \\
        \hline
    \end{tabular}
    \label{tab:dataExperiments}
\end{table}

The results show, that the \GSM algorithm runs approximately $23\%$ faster than \SMO (ignoring the fact
that \SMO is faulty by design). Also, the performance of \GSM and \BPCS is almost indistinguishable
and according to our experiments, it varies in the span of units of percents depending on the exact CPU,
cache, RAM and compiler setting.
The seemingly superior average performance of \BPBCS is caused by the heuristics \BPBCS uses;
however, while the worst-case performance of GSM is guaranteed, the performance of \BPBCS for certain patterns is
worse than that of \GSM.
Also note that \GSM is a streaming algorithm while the others are not.

\begin{table}[t]
    \centering
    \caption{Found occurrences across datasets: The value is simply the sum of occurrences over all the patterns.}
    \setlength{\tabcolsep}{4pt}
    \begin{tabular}{|c|c|c|c|c|}
        \hline
        \multirow{2}{*}{Algorithm}&\multicolumn{3}{|c|}{Dataset}\\
        \cline{2-4}
        & CH & HS & BIB \\
        \hline
        \SMAL & 86243500784 & 51136419 & 315612770 \\
        rest & 84411799892 & 51034766 & 315606151 \\
        \hline
    \end{tabular}
    \label{tab:flawExperiments}
\end{table}

Table \ref{tab:flawExperiments} visualizes the accurateness of \SMI with respect to its flaw by comparing the number of occurrences found by the respective algorithms.
The ratio of false positives to true positives for the \SMI was: CH $2.17\%$, HS $0.20\%$ and BIB $0.002\%$.

}

\bibliographystyle{splncs03}
\bibliography{main}

\end{document}